\documentclass[a4paper,11pt]{article}
\usepackage[titletoc]{appendix}

\usepackage[utf8]{inputenc}
\usepackage{amsfonts}
\usepackage{amsmath}
\usepackage{amssymb}
\usepackage{amsthm}
\usepackage{url}
\usepackage[english]{babel}
\usepackage{verbatim}
\usepackage{graphicx}
\usepackage{enumerate}

\usepackage{hyperref}

\makeatletter
\def\namedlabel#1#2{\begingroup
    #2%
    \def\@currentlabel{#2}%
    \phantomsection\label{#1}\endgroup
}
\makeatother

\numberwithin{equation}{section}

\def\argmax{{\rm{arg}}\max}

\def\defto{\buildrel def\over =}

\def\esssup{\hbox{ess sup}}
\def\half{\frac{1}{2}}

\def\mp{\mu_+}
\def\mm{\mu_-}
\def\sp{\sigma_+}
\def\sm{\sigma_-}

\def\g{{g_{\epsilon}}}

\def\cL{{\mathcal L}}
\def\cLp{{\cL^+}}
\def\cLm{{\cL^-}}

\def\tC{{\mathbf{C}}}
\def\C{\tC}

\def\F{{\mathcal F}}

\def\R{{\mathbb R}}

\def\E{{\mathbb E}}

\def\val{v}
\def\Val{{\mathbb V}}
\def\snell{{\mathbf V}}
\def\s2{{\mathbf U}}
\def\gain{{g}}
\def\Gain{{\mathbf G}}
\def\vg{{u}}

\def\inclm{{\phi^{L,M}}}
\def\declm{{\psi^{L,M}}}
\def\tmh{{\tau^{\hat m}_{H}}}
\def\hm{{\hat m}}

\def\ap{{\alpha_+}}
\def\am{{\alpha_-}}
\def\bp{{\beta_+}}
\def\bm{{\beta_-}}
\def\fp{{h_+}}
\def\fm{{h_-}}

\def\k{{f_{\kappa}}}
\def\di{{D}}

\def\he{the trader  }
\def\He{The trader}
\def\his{their }

\def\endpf{{\hfill$\diamond$}}

\theoremstyle{plain}
\newtheorem{theorem}{Theorem}[section]
\newtheorem{lemma}[theorem]{Lemma}

\newtheorem{example}[theorem]{Example}
\theoremstyle{remark}
\newtheorem{remark}{Remark}
\theoremstyle{definition}

 \title{An Optimal Stopping Problem Modeling Technical Analysis}
\usepackage{authblk}

\author[1,2]{Saul Jacka\thanks{Corresponding author. Saul Jacka gratefully acknowledges funding received from the EPSRC grant EP/P00377X/1 and is also grateful to the Alan Turing Institute for their financial support under the EPSRC 
grant EP/N510129/1.}
\thanks{Both authors would like to thank the participants in the 3rd Conference on Mathematical Economics and Finance at the Centre for Mathematical Research in Economics and Finance for many helpful comments}
\thanks{ E-mail: \textit{s.d.jacka@warwick.ac.uk}}}
\author[1]{Jun Maeda \thanks{E-mail: \textit{j.maeda@warwick.ac.uk}\\Department of Statistics\\University of Warwick\\Coventry CV4 7AL, UK}}

\affil[1]{University of Warwick}
\affil[2]{Alan Turing Institute}

\begin{document}
\maketitle
   \begin{abstract}
We present a solution to an optimal stopping problem for a process with a wide-class of novel dynamics. The dynamics model the support/resistance line concept from financial technical analysis.
\end{abstract}




\section{Introduction and motivation}
\label{section: Introduction}
Many financial traders base their trading strategies on technical analysis (TA). The analysis relies heavily on the visual shape of historical price graphs (which traders call \lq charts') to determine whether the asset is a good buy or not. One of the basic modes of analysis in the field is that of a support and resistance line. In this method, the traders obtain a horizontal line called a support (respectively, resistance) line that they believe is a local support (ceiling) of the asset price. The belief is that if the stock price crosses a support line from above and goes lower than the level by a significant amount, then the stock has moved to a  regime with negative outlook, in which case traders should sell, or, at least, not be long of the stock. Conversely, if the asset price spikes up, crossing a resistance line from below, the asset is considered to have shifted to a regime with positive outlook and the traders should buy the asset or at least cover any short position.

We note here that the support/resistance level is not a hard limit. So, the stock can go lower (higher) than the support (resistance) level without the regime changing and is expected to come back up (down) to the relevant side in a short period of time. We also note that there may be several support/resistance levels in one chart. 

Studies on TA have been performed, but they mainly focus on how to detect the sign of the regime transition as quick as possible and  on checking against historical data (usually by means of computation and statistics) the usefulness of adopting TA in trading. Some examples of research that focus on these points are \cite{bdgtt07} and \cite{lmw00}. We know of no literature attempting to model and justify TA methods mathematically.

In this paper,  we assume that there are only two regimes for the stock price. Under these regimes the price follows different diffusion dynamics (which we will, from time to time, specialise to  log-normal dynamics\footnote{An earlier version of this paper appeared on arXiv and only considered the log-normal case.}). This introduces a novel class of Markovian price dynamics. We then define criteria for buying/selling the stock by solving two optimal stopping problems: the first being the optimal selling problem and the second the buying problem. Of course, we suppose that purchase precedes selling temporally but this means, since we need to solve the problems iteratively, that we have first to solve the selling problem in order to specify the buying problem correctly.

One of the features that makes our model special is that the two regimes are not spatially distinct, i.e. there is a region 
where the stock price can be in either of the two regimes. This feature provides some ``room'' for the process in each 
regime to move around the support/resistance level without switching to the other regime.

The rest of the paper is organized as follows: Section~\ref{section: setup} presents the setup we use for the model with a support/resistance level. We solve for the optimal selling time (which is of stop-loss type) that maximizes the expected discounted price in Section 3. In Section~\ref{section: buying}, we solve the optimal stopping problem for purchasing the shares. We make some suggestions for further modelling in Section~\ref{section: conclusions}.

\section{The price model and the selling problem}
\label{section: setup}
\subsection{The model}
We assume that there are price levels $L$ and $H$ ($0<L<H$) at which the regimes change. The positive regime corresponds to the interval $[L, \infty)$ and the negative regime to  $(0, H]$.  We may think of the support/resistance level as situated somewhere in $(L,H)$, say at $\frac{L+H}{2}$.

The dynamics for the stock price are expressed by the following SDEs, one for each regime:

\begin{equation}\label{stock}
\begin{cases}
dS^+_t = \mu_{+} (S^+_t)dt + \sigma_{+}(S^+_t)dW_t \quad \text{in the positive regime},\\
dS^-_t = \mu_{-} (S^-_t)dt + \sigma_{-}(S^-_t)dW_t \quad  \text{in the negative regime}, 
\end{cases}
\end{equation}
where  $\sigma_{+}$, $\sigma_{-}$, $\mu_{+}$, $\mu_{-}$ are Holder-continuous functions with $\sp$ and $\sm$ strictly positive away from 0,  and $W_t$ is a one dimensional Brownian motion. Consistent with the modeling of a stock {\em price}, we assume that zero is either absorbing or inaccessible for a  process following  the negative-regime dynamics.We denote the associated differential operators (the restriction of the associated infinitesimal generators to $\C^2$ functions) by $\cLp$ and $\cLm$, so that
$$\cLp:f\mapsto \half \sp^2 f''+\mp f'\text{ and }\cLm:f\mapsto \half \sm^2 f''+\mm f'.
$$

Let $r$ denote the risk-free interest rate, then we assume that
\begin{equation}\label{cond}
\mu_{-}(x)  \leq rx \leq \mu_{+}(x) \text{ and }r\geq 0.
\end{equation}

We define the c\`adl\`ag flag-process $F_t$  taking values in $\{-, +\}$ as

\begin{equation}
F_t = 
\begin{cases}
+ \quad \text{when the stock is in the positive regime}\\
- \quad \text{in the negative regime}.
\end{cases}
\end{equation}
Thus
$F_t$ jumps from one value to the other only in the following cases:

\begin{equation}
\begin{cases}
\text{if }F_{t-} =+ \text{ and } S_t = L \text{, then } F_t = -\\
\text{if }F_{t-} = - \text{ and } S_t = H \text{, then } F_t = +,
\end{cases}
\end{equation}
and the stock price satisfies
\begin{equation}\label{stock2}
dS_t=\mu_{F_t}(S_t)dt+\sigma_{F_t}(S_t)dW_t.
\end{equation}
Notice that the separation of $L$ and $H$ ensures that $F$ only has finitely many jumps on any finite time-interval and that the pair $(S_t,F_t)$ clearly is a Feller process.

We further assume that the trader {\em always sells \his shares at the level $M\geq H$}, either because  \he is happy to take profit or is required to do so by \his manager. 

\begin{remark}
It is not difficult to understand why a professional trader should take a profit at some level $M$, since \he needs to choose an investment from a collection of stocks. If \he wants to pick one from the group, not only does \he compare the possible losses, but also the possible profits in investing in the stocks and will normally wish to close-out sufficiently profitable positions. 
\end{remark}

Note that the two regimes have non-empty intersection, $[L, H]$, of their domains and condition \eqref{cond} implies that the discounted price process has supermartingale dynamics in the negative regime and submartingale dynamics in the positive regime.
\subsection{Some notation and further assumptions}
We denote generic exit/entrance times by $\tau$ so that 
\begin{itemize}
\item the first entrance by $S$ to the set $I$ is denoted $\tau_I$: $\tau_I=\inf\{t:\; S_t\in I\}$ 
\item similarly, the first entrance by $S$ to the set $I$, when the flag value is $f$, is denoted $\tau_{I,f}$: $\tau_{I,f}=\inf\{t:\; S_t\in I\text{ and }F_t=f\}$ 
\end{itemize}
We recall that there is a unique in law solution to the stock-price dynamics equations (\ref{stock}) and hence to (\ref{stock2}) (\cite{ks98}). Moreover, there are unique fundamental solutions on $\R_+$, which we denote $\phi_+$, $\psi_+$, respectively $\phi_-$, $\psi_-$  to the ODEs
\begin{equation}\label{fund}\cLp f-rf=0
\end{equation}
and 
\begin{equation}\cLm f-rf=0
\end{equation} 
satisfying:

\begin{equation}
\phi_+(0)=\phi_-(0)=0;
\end{equation}
\begin{equation}\phi_+(H)=\phi_-(H)=1;
\end{equation}
\begin{equation}\psi_+(L)=\psi_-(L)=1;
\end{equation}
\begin{equation}\label{fund2}
\lim_{x\rightarrow\infty}\psi_+(x)=\lim_{x\rightarrow\infty}\psi_-(x)=0
\end{equation}
(see, for example, \cite{Friedman}).

\subsection{The selling problem} If we assume that the trader already holds the stock, then if they wish to maximise their expected discounted profit they will seek the optimal time to sell. So they will seek a stopping time $\tau$, bounded by $\tau_{[M,\infty)}$, which achieves
$$
\sup_{\tau\leq \tau_{[M,\infty)}}\E[e^{-r\tau}S_\tau].
$$
For each initial price $x\in(0,M]$ and flag-value $f\in\{+,-\}$, we define
\begin{equation}\label{prob1}
\snell(x,f)=\sup_{\tau\leq \tau_{[M,\infty)}}\E_{x,f}[e^{-r\tau}S_\tau].
\end{equation}

\begin{remark}If, rather than condition (\ref{cond}), we suppose that $\mu_-(x)\leq \mu_+(x)\leq rx$ then the discounted stock price is actually a supermartingale bounded above by $M$ and it follows immediately from the Optional Sampling Theorem that it is always optimal to sell the stock immediately. Conversely, if $rx\leq \mu_-(x)\leq \mu_+(x)$ then the discounted stock price is actually a submartingale bounded above by $M$ and now it is always optimal (by the Optional Sampling Theorem) to wait until the last possible time, $\tau_{[M,\infty)}$, before selling.
\end{remark}
In the case where (\ref{cond}) holds, so the two drifts sandwich the return from a risk-free asset, we expect the possibility of an earlier sale. As we shall see, the trader should never sell in the positive regime unless the stock price has attained level $M$ (thus earlier sale always corresponds to a ``stop-loss'' action). 

We will show that the optimal action is to sell at $\tau^{\hat m}_M\defto \tau_{[M,\infty)}\wedge \tau_{[0,\hat m],-}$, the earlier of $\tau_{[M,\infty)}$ and $\tau_{[0,\hat m],-}$, for some ${\hat m}\leq H$. 

\subsection{ The form of the solution}
Our analysis will consider decreasing possible values, $m$, for the selling boundary $\hat m$ and will divide into three cases:  
\begin{enumerate}[(C1)]
\item where $\hat m\in [L, H]$, 
\item where $\hat m\in(0,L)$ and 
\item where $\hat m=0$. 
\end{enumerate}
In case C1, notice that if $F_0=+$ then the stock will be sold as soon as the price hits $L$. Recalling that the optimal future payoff (Snell envelope) for an optimal stopping problem is a martingale up until the (last) optimal stopping time, it follows that for such a value of $m$,
$e^{-r(t\wedge \tau_{[L,M]^c,+})}\snell(S_{t\wedge \tau_{[L,M]^c,+}},F_{t\wedge \tau_{[L,M]^c,+}})$ should be a martingale with
\begin{equation}\label{bound1p}
\snell(L,+)=L\text{ and }\snell(M,+)=M.
\end{equation}
Similarly, if $F_0=-$, $e^{-r(t\wedge \tmh)}\snell(S_{t\wedge \tmh} ,F_{t\wedge \tmh})$  should be a martingale  with
\begin{equation}\label{bound1m}
\snell(H,-)=\snell(H,+)\text{ and }\snell(m,-)=m,
\end{equation}
for the optimal choice of $m$.

In case C2, if $F_0=+$,  $e^{-r(t\wedge \tau_{[L,M]^c,+})}\snell(S_{t\wedge \tau_{[L,M]^c,+}},F_{t\wedge \tau_{[L,M]^c,+}})$ should again be a martingale, but now with boundary conditions
\begin{equation}\label{bound2p}
\snell(L,+)=\snell(L,-)\text{ and }\snell(M,+)=M
\end{equation}
and
\begin{equation}\label{bound2m}
\snell(H,-)=\snell(H,+)\text{ and }\snell(m,-)=m;
\end{equation}
while in case C3 the requirements are the same as in case C2, with $\hm=0$, the boundary condition at 0 corresponding either to the inaccessibility of 0 or to the the fact that it is an absorbing boundary.

Standard arguments (see e.g. \cite{ks98} etc.)  tell us that the unique solution to the first characterisation is given by $v(x,f,m)$ where $v(x,+,m)$ satisfies
\begin{equation}\label{odep}
\cLp v-rv=0
\end{equation} 
and with boundary conditions (\ref{bound1p}) and (then) $v(x,-,m)$ satisfies
\begin{equation}\label{odem}
\cLm v-rv=0
\end{equation} 
with boundary conditions (\ref{bound1m}).
In a similar fashion, by taking  $v(H,+,m)=\E_{H,+}[e^{-r\tau^{\hat m}_M}S_{\tau^{\hat m}_M}]$ and $v(L,-,m)=\E_{L,-}[e^{-r\tau^{\hat m}_M}S_{\tau^{\hat m}_M}]$ and solving (\ref{odep}) and (\ref{odem}) with these values in (\ref{bound2p}) and (\ref{bound2m}) we obtain the unique solution to the second characterisation.

For the third case, the usual arguments show that, $\phi_-$ given by (\ref{fund})-(\ref{fund2}) satisfies
\begin{equation}\label{phi}\phi_-(x)=\E_{x,-}[e^{-r\tau_H}].
\end{equation}

Similarly, if we define $\inclm$ and $\declm$ by 
$$
\inclm(x)=\E_{x,+}[e^{-r\tau_{[M,\infty)}}1_{(\tau_{[L,M]^c,+}=\tau_{[M,\infty)})}]\text{ and }\declm(x)= \E_{x,+}[e^{-r\tau_{[M,\infty)}}1_{(\tau_{[L,M]^c,+}=\tau_{[0,L],+})}],
$$
then $v(x,-,0)=\phi_-(x)v(H,+,0)$ and
$$
v(H,+,0)=M\inclm(H)+\val(L,-)\declm(H),
$$
so that 
$$
v(H,+,0)=\frac{M\inclm(H)}{1-\phi_-(L)\declm(H)}.
$$

 Consequently, if we take 
 \begin{equation}\label{disp}
 v(x,-,0)=\phi_-(x)\frac{M\inclm(H)}{1-\phi_-(L)\declm(H)}
 \end{equation}
 and $v(x,+,0)$ as the unique solution to (\ref{odep}) with 
 boundary conditions (\ref{bound2p}), we obtain the unique solution to the third characterisation.
 
 \subsection{Identifying the stop-loss boundary}
 To identify the stop-loss boundary, $\hm$, we consider the argument we will make to show that our proposed solution is optimal. We will adopt the usual techniques, using the characterisation of the Snell envelope which works as follows: we will show that our candidate solution $\val$ has the following properties:
 \begin{enumerate}[(P1)]
 \item $\Val_t\defto e^{-rt}\val(S_{t\wedge \tau_{[M,\infty)}},F_{t\wedge \tau_{[M,\infty)}})$ is a class D supermartingale;
 \item $\Val_t$ dominates the gains process $e^{-rt}S_{t\wedge \tau_{[M,\infty)}}$;
 \item there exists a stopping time $\hat\tau\leq \tau_{[M,\infty)}$ such that $\Val_0=\E[e^{-r\hat\tau}S_{\hat\tau}]$.
 \end{enumerate}
This is sufficient to show that $\val$ is the optimal solution and that $\hat\tau$ is an optimal stopping time.

Now, to show this we'll characterise $\hat m$ as follows: in cases C1 and C2 we'll take $\hat m$ to be the unique choice of $m$ for which 
\begin{equation}
\frac{\partial \val}{\partial x}(m,-,m)=1,
\end{equation} while in the third case we'll show that $\frac{\partial \val}{\partial x}\geq 1$ for all $x>0$ (see Theorem \ref{sell}).

We'll establish that this is sufficient, since if we define $g(x,f)\defto \val(x,f)-x$ then the strong maximum and strong minimum principles will tell us that $g\geq 0$, while smooth pasting at $\hat m$ (between $\val$ and the identity) gives us that $\Val$ is a supermartingale, while it is clearly class D since bounded (by $M$).

Case C3 is similar except that we don't need smooth pasting when $\hat m=0$.


\section{Establishing the solution of the selling problem}
\subsection{The general case}
\begin{theorem}\label{sell}
Under the model outlined in section 2, there are three possibilities corresponding to cases C1 to C3; so exactly one of the following holds:
\begin{itemize}
\item[]letting $\val_1(x,f)$ be the solution to the linked ODEs (\ref{odep}) and (\ref{odem}) with boundary conditions (\ref{bound1p})
and (\ref{bound1m}); then 
\item[(S1)]{\em either} there exists $\hat m\in[L,H]$ such that $\frac{\partial \val_1}{\partial x}(x,-;m)|_{x=m=\hat m}=1$
\item[]{\em or}

\item[]  for every $m\in [L,H],$ $\frac{\partial \val_1}{\partial x}(x,-;m)|_{x=m}>1$.
\item[] In the latter case,  defining $\val_2(x,f;m)$ to be the solution to the ODEs (\ref{odep}) and (\ref{odem}) with boundary conditions (\ref{bound2p})
and (\ref{bound2m}), 
\item[(S2)]{\em either} there exists an $\hat m\in(0,L)$ such that $\frac{\partial \val_2}{\partial x}(x,-;m)|_{x=m=\hat m}=1$ 
\item[]{\em or}
\item[(S3)]  for every $m\in (0,L)$, $\frac{\partial \val_2}{\partial x}(x,-;m)|_{x=m=}>1$ 
\end{itemize}
\end{theorem}
\begin{proof} We claim that $\hat m$  is well-defined, so that the three cases are exhaustive.

First, define $\phi_-$ as in (\ref{phi});  $\psi_-$ as the unique (and decreasing) solution on $(0,H)$ of $\cL^- f=0$ with $\psi_-(H)=0$ and $\psi_-(L)=1$; $\phi_+$ as the unique, increasing, solution of $\cL^+ f=0$ with $\phi_+(M)=1$ and $\phi_+(L)=0$; and $\psi_+$ as the unique, decreasing solution of  $\cL^+ f=0$ with $\psi_+(M)=0$ and $\psi_+(L)=1$. 
Our solutions  $\val_1$, $\val_2$ and $\val_3$ to (\ref{odep}) and (\ref{odem}) will be of the form
$\val(x,f;m)=A_f(m)\phi_f+B_f(m)\psi_f$ where, setting $C(m)=(A_-(m),B_-(m),A_+(m)B_+(m))^T$, for suitable choices of the coefficients, it follows from the boundary conditions that

\begin{equation}
N(m)C(m)=(m,m-L,0,M)^T
\end{equation}
in  case S1
and
\begin{equation}
\tilde N(m)C(m)=(m,0,0,M)^T
\end{equation}
in cases S2 and S3, where
\begin{equation}
N(m)=\begin{pmatrix}
\phi_-(m)&\psi_-(m)&0&0\\
\phi_-(m)&\psi_-(m)&\phi_+(L)&\psi_+(L)\\
\phi_-(H)&\psi_-(H)&\phi_+(H)&\psi_+(H)\\
0&0&\phi_+(M)&\psi_+(M)\\
\end{pmatrix}
\end{equation}
and
\begin{equation}
\tilde N(m)=\begin{pmatrix}
\phi_-(m)&\psi_-(m)&0&0\\
\phi_-(L)&\psi_-(L)&\phi_+(L)&\psi_+(L)\\
\phi_-(H)&\psi_-(H)&\phi_+(H)&\psi_+(H)\\
0&0&\phi_+(M)&\psi_+(M)\\
\end{pmatrix}.
\end{equation}
It follows fairly easily from the Implicit Function Theorem that $\frac{\partial v}{\partial x}(x,-;m)$ is jointly continuous in $(x,m)$.

Now take $L<m<H$, so that $\val(x,+;m)\geq x$ on $[L,M]$ and so $\val(m,-;m)=m$ and $\val(H,-;m)=\val(H,+;m)\geq H$. This implies that there is a $\theta\in(0,1)$ such that  $\frac{\partial v}{\partial x}(m+\theta (H-m),-;m)\geq 1$ and letting $m\uparrow H$ we see that 
 $\frac{\partial v}{\partial x}(H,-;H)\geq 1$. Then the result follows by continuity.
\end{proof}

\begin{theorem}\label{sell2}
For the three cases:
\begin{itemize}
\item In case S1, 
\begin{equation}
\val(x,f)=\begin{cases}
\val_1(x,f;\hat m):\; x\in[\hat m, H], f=-\text{ or }x\in [L,M],f=+\\
x:\; x\in[0,\hat m), f=- \text{ or }x\in[0,L], f=+;
\end{cases}
\end{equation}

\item In case S2,
\begin{equation}
\val(x,f)=\begin{cases}
\val_2(x,f;\hat m):\; x\in[\hat m, M]\\
x:\; x\in[0,\hat m);
\end{cases}
\end{equation}

\item In case S3,
\begin{equation}
\val(x,f)=\val_3(x,f)
\end{equation}
where $\val_3$ is the solution to the ODEs (\ref{odep}) and (\ref{odem}) with boundary conditions (\ref{bound2p})
and satisfying (\ref{disp}).
\end{itemize}
\end{theorem}
We outline the proof here, relegating some details to the appendix.
\begin{proof}First we require
\begin{lemma}\label{snell}
Suppose our candidate optimal value function, given by the trichotomy S1-S3, and denoted by $\val$, satisfies properties P1-P3 then it is optimal
\end{lemma}
This is standard, see the appendix for the proof.

So we seek to prove that $\val$ has properties P1-P3.

To prove P3: notice that each $\val_i$ is continuous on $[0,M]$ so bounded and by the usual arguments is
actually given by $\E_{x,f}[e^{-r\tau^{\hm}_M})S_{\tau^{\hm}_M}]$, 
so that $\hat \tau=\tau^{\hm}_M$.

To prove P1: since $\Val$ is bounded it is definitely of class D. Since $L$ is strictly less than $H$, we have, from the It\^o-Tanaka formula,
\begin{eqnarray}\label{super}
d\Val_t=1_{(t< \tau_{[M,\infty)})}e^{-rt}\biggl[\biggl(&-r\val(S_t,F_t;\hat m)+\cL^+\val(S_t,+)1_{(F_t=+)}+L^-\val(S_t,+)1_{(F_t=-)}\biggr)dt\nonumber\\
&+dM_t+\bigl(\frac{\partial \val}{\partial x}(\hat m, -;\hat m)-\frac{\partial \val}{\partial x}(\hat m -,-;\hat m)\bigr)dl^{\hat m}_t\biggr],
\end{eqnarray}
where $M$ is a continuous local martingale, $l^{\hat m}$ is the local time of $S$ at $\hat m$, and the last term disappears in case C3 since $0$ is either absorbing or inaccessible.
Now since $\val(x,-;\hat m)=x$ for $x\leq \hat m$ and since (in cases C1 and C2) we have imposed the condition that $\frac{\partial \val}{\partial x}(\hat m, -;\hat m)=1$ we see that the local time term in (\ref{super}) disappears. Then, thanks to (\ref{odep}) and (\ref{odem}) the other bounded variation terms in (\ref{super}) disappear when $F_t=+$ and when $F_t=-$ and $S_t\geq \hat m$, so we are left with
\begin{equation}
d\Val_t=1_{(t< \tau_{[M,\infty)})}e^{-rt}\biggl[(\mu_-(S_t)-rS_t)1_{(F_t=-,S_t<\hat m)}dt+dM_t\biggr],
\end{equation}
so that $\Val$ is a local supermartingale (since $\mu_-(x)\leq rx$). Then, since $\val$ is bounded it follows that $\Val$ is a supermartingale as required.
Finally, to establish P2: we need to prove that $\val\geq x$.
This is Lemma \ref{dom} in Appendix \ref{app}.
\end{proof}

\subsection{The lognormal case}
In this subsection, we assume that the price dynamics are lognormal in each regime so that
\begin{equation}\label{log}
\sp(x)=\sp x;\; \sm(x)=\sm x;\; \mp(x)=\mp x;\; \mm(x)=\mm x,
\end{equation}
with
$$\sp>0;\; \sm>0; \text{ and }\mm<r<\mp.
$$
The corresponding fundamental solutions of (\ref{odep}) and (\ref{odem}) are given by $x^\ap$, $x^{\bp}$ and $x^\am$, $x^\bm$ respectively, where $\ap$ and $\bp$ are the roots of $\fp(t)\defto\half\sp^2t^2+(\mp-\half\sp^2)t-r=0$ and $\am$ and $\bm$ are the roots of  $\fm(t)\defto\half\sm^2t^2+(\mm-\half\sm^2)t-r=0$, ordered so that $\beta_{\pm}<\alpha_{\pm}$.
\begin{lemma}
If the price dynamics correspond to (\ref{log}) then case C3 does not occur.
\end{lemma}
\begin{proof}
Notice that  $\fm(0)=-r\leq0$ and $\fm(1)=\mm-r<0$. Since $\sm^2>0$ it follows that $\bm\leq 0<1<\am$.

Consequently $\phi_-(x)=(\frac{x}{H})^\am$ and $\val_3(x,-;0)=cx^\am$ for some positive constant $c$. Thus 
$\frac{\partial \val_3}{\partial x}(0,-;0)=0$ and so $\hat m>0$.

\end{proof}
\begin{example}\label{ex}
Taking the dynamics of (\ref{log}), set $r=.02$, $\sp^2=.06$, $\mp=.04$, $\sm^2=.01$ and $\mm=.005$.

The general solution to (\ref{odem} )is $Ex^{2}+Fx^{-2}$ and taking $\hat m=1$, the solution $v$ with  $v(1)=1$ and $v'(1)=1$ is $v(x)=\frac{3}{4}x^2+\frac{1}{4}x^{-2}$. Taking $H=2$, we see that $v(H)=\frac{49}{16}$.

Now solving (\ref{odep}) for $v(x,+)$ we get the general solution $Cx^{\frac{2}{3}}+Dx^{-1}$. Imposing the conditions that $v(L,+)=L$ with $L=\frac{1}{8}$, and $V(H,-)=v(H,+)$ we obtain 
$C=\frac{391}{2(64.2^{\frac{2}{3}}-1)}\sim 1.943462$ and $D=\frac{49}{8}-2^{\frac{5}{3}}C\sim -.045108$.

Now solving for $v(M,+)=M$ gives $M$ as the root of $CM^{\frac{5}{3}}+D-M^2=0$, so $M\sim 7.322007$.

Thus, with these values for $L,M$ and $H$ we get $\hat m=1$ and $v$ is as above.
\end{example}

\section{Optimal Timing of Purchase}
\label{section: buying}

 We now consider the optimal time to purchase the stock. 

The gains function $\gain$ is given by $\gain:(x,f)\mapsto \val(x,f)-x$ and our gains process is $\Gain$, given by
\begin{equation}
\Gain_t = e^{-rt}\gain(S_t, F_t).\\
\end{equation}
We seek to find:
$$\s2_t\defto\esssup_{\tau\geq t}\E[\Gain_{\tau}|\F_t].
$$

The optimal stopping problem corresponds to the case when the buyer pays interest (or at least incurs a notional opportunity cost of interest foregone) on the purchase price from the time of purchase and seeks to maximise their profit. 
\begin{remark}
There are other possibilities for the gains in the buying problem, such as proportional reward, where $g=\val/x$. This corresponds to maximising profit per unit of expenditure.
\end{remark}

\begin{theorem}\label{buy}
It is optimal to purchase the shares only when the underlying process is in the positive regime. In this case, there is an optimal level $B\in[L,M)$, given by $B=\argmax_{x\in [L,M]}[\frac{\gain(x,+)}{\psi_+(x)}]$, such that it is optimal to buy if and only if the stock has a price in $[L,B]$
\end{theorem}

\begin{proof}
It follows from the proof of Theorem \ref{sell2} that, since $\cL^-\val(x,-)-r\val(x,-)=0$ on $(\hm,H)$,
$$
\cL^-\gain(x,-)-r\gain(x,-)=(rx-\mm(x))1_{(\hm,H)}(x)\geq 0,
$$
and $g(x,-)$ is $C^1$ and piecewise $C^2$ on $(0,H)$. We deduce that, defining $\tau_{[H,\infty)}$ as the first time that the stock enters the positive regime, $\Gain_{t\wedge \tau_{[H,\infty)}}$ is a submartingale. Consequently, it is always optimal to continue (i.e. not purchase) when the stock is in the negative regime. 

Now we define  $\rho$ by  setting $\rho=\frac{g(\cdot,+)}{\psi_+}$ and let $B$ be the $\argmax$ of $\rho$ on $[L,M]$. Notice that $B<M$ since $g\geq 0$, $g(M,+)=0$ and  $g(\cdot,+)>0$ on $(L,M)$. 

Now define $\vg$ by 
$$\vg(x,+)=\begin{cases}\gain(x,+): \text{ if }x\leq B\\
\frac{\psi_+(x)}{\psi_+(B)}\gain(B,+):  \text{ if }x\geq B
\end{cases}
$$
and
$$
\vg(x,-)=\phi_-(x)\vg(H,+).
$$
This corresponds to the expected gain from buying at the time stated (since $\phi_-(x)=\E_{x,-}[e^{-r\tau_H}]$ and $\frac{\psi_+(x)}{\psi_+(B)}=\E_{x,+}[e^{-r\tau_{[L,B],+}}]$).

Thus, applying the arguments from the proof of Theorem \ref{sell2}, we see that we need only prove that
\begin{enumerate}[(1)]
\item $\vg\geq \gain$
\item $\cL^+\vg(\cdot,+)-r\vg(\cdot,+)\leq 0 \text{ on }(L,M)$
and
\item $\vg_x(B,+)-\gain_x(B,+)\leq 0$
\end{enumerate}
We have already proved that $\cL^+\val(\cdot,+)-r\val(\cdot,+)= 0 \text{ on }(L,M)$ in the proof of Theorem \ref{sell2}. Thus $\cL^+\gain(\cdot,+)-r\gain(\cdot,+)=rx-\mu_+(x)\leq 0$ on $(L,B)$.
 Property (2) follows, since $\cL^+ \psi_+-r\psi_+=0$. 
 
To establish Property (1), first notice that $u(\cdot,+)=\gain(\cdot,+)$ on $[L,B]$ and $[M,\infty)$. On the interval $[B,M]$, $u(x,+)-\gain(x,+)=\psi_+(x)\frac{\gain(B,+)}{\psi_+(B)}-\gain(x,+)$, which is non-negative by the definition of $B$ and the positivity of $\psi_+$. 

Now define $\di$ by $\di:x\mapsto \vg(x,-)-\gain(x,-)$. Then, since $\cL^-\di-r\di=
-(\cL^-\gain(x,-)-r\gain(x,-))=-(rx-\mm(x))1_{(\hm,H)}(x)\leq 0,$ it follows that $e^{-r(t\wedge \tau_{[H,\infty)})}(\di(S_{t\wedge \tau_{[H,\infty)}})$ is a bounded supermartingale with terminal value $\di(H)1_{\tau_{[H,\infty)}<\infty}$. Since $\di(H)=\vg(H,-)-\gain(H,-)=\vg(H,+)-\gain(H,+)$, and we have already shown that this is non-negative, it follows from the Optional Sampling Theorem that $\di(S_0)=\vg(S_0,-)-\gain(S_0,-)\geq 0$ for each choice of $S_0$ in $[0,H]$.

To establish 
Property (3), take $\kappa\geq 0$ and consider $\k$ given by $\k:x\mapsto \gain(x,+)-\kappa \psi_+(x)$. Since $\psi_+$ satisfies (\ref{odep}) it follows that
$$
\cL^+\k-r\k= rx-\mp(x)\leq 0,
$$
so that $\k$ satisfies the strong minimum principle on $(L,M)$. In particular, defining $\rho$ by  setting $\rho=\frac{g(\cdot,+)}{\psi_+}$ and taking $\kappa=\rho(L)$ we see that $\k(L)=0\geq \k(M)$ and so either $\k$ increases to a unique and strictly positive maximum in $(L,M)$ or it is monotone decreasing on $[L,M]$.

Property (3) then follows from our characterisation of $B$, since $B>L$ implies that $\rho'(B)=0$ which, in turn, implies that $\frac{\psi_+'(B)}{\psi_+(B)}=\frac{\gain_x(B,+)}{\gain(B,+)}$, while $B=L$ implies that $\k$ is monotone decreasing on $[L,M]$ and so $\k'(L)=\gain_x(L,+)-u_x(L,+)\leq 0$.
\end{proof}

\begin{example}
If we return to the set-up of Example \ref{ex}, we see that $\gain(x,+)=Cx^{\frac{2}{3}}+Dx^{-1}-x$, where 
$C$ and $D$ are given there, while $\psi_+(x)=2x^{-1}$. It follows that $\rho(x)=\frac{C}{2}x^{\frac{5}{3}}-\half x^2+\half D$.Differentiating, we see that $\rho$ is maximised at $B=(\frac{5C}{6})^3\sim  4.248001\in (L,M)$ and thus $B>L$.
\end{example}

\section{Concluding remarks}
\label{section: conclusions}
Two other possibilities for modelling the price dynamics in the presence of  support/resistance levels, suggest themselves. One is to use stochastic delay differential equations (SDDEs; \cite{by11}, \cite{longtin10}, \cite{mao97}, \cite{ymy08}). This makes some sense, since TA is the method traders use to forecast dynamics of the future stock price from analysis of historical prices, and SDDEs are stochastic differential equations (SDEs) with coefficients that depend on historical levels. 

The other is to model the stock price (at least locally near the support line) as a skew Brownian motion (see \cite{im65}). Using this process to describe the underlying stock price process under our setup requires less parameters than using SDDEs. However, as the papers of \cite{ns11} and \cite{rossello12} show, the model with skew Brownian motion has arbitrage opportunities, which is undesirable in a financial context.

We hypothesise that our solution could be extended to this latter setting, at least if other dynamics remained the same: so we would assume partial reflection upwards at the level $R\in (L,M)$ when the stock is in the positive regime and partial reflection downwards at $R$ when the stock was in the negative regime. The corresponding dynamics would have generators $\cLp$ and $\cLm$ with scale  measures with a \lq\lq kink'' at $R$---upwards in the case of the positive regime and downwards for the negative regime (see\cite{im65} and \cite{hs81}). 

We do not seek to analyse this case further here.

 \appendix
\section{Appendix}\label{app}
{\em Proof of Lemma \ref{snell}
{\rm Take the joint price and flag process started at $(x,f)$ and consider the corresponding process $\Val_t$. By P1 and the Optional Sampling Theorem for Class D supermartingales,
for any stopping time $\tau$,
$$v(x,f)=\Val_0\geq \E_{x,f}[e^{-r\tau}\val(S_\tau)]\geq \E_{x,f}[e^{-r\tau}S_\tau] \text{ (the last inequality follows by P2)}
$$
and it follows that $\val\geq V$.
Conversely, by P3,
$$
\Val_0=\E_{x,f}[e^{-r\hat\tau}S_{\hat\tau}]
$$
and so $\val\leq V$.
\endpf}

\begin{lemma}\label{dom}
The function $\Val(x,+)\geq x$ for all $x\in [L,M]$ while $\Val(x,-)\geq x$ for all $x\in [0,H]$. 

\end{lemma}
\begin{proof}
As indicated in section 2, the main tool here is the strong maximum/minimum principle (see \cite{Friedman} or \cite{ps07}).
First, define $g(x,f)\defto \Val(x,f)-x$. Then 
$$
\cL^+ g(x,+)-rg(x,+)=rx-\mu_+(x)\leq 0 \text{ on }[L,M].
$$
with $g(M,+)=0$ and $g(L,+)=g(L,-)$.
Now the strong minimum principle tells us that $g(\cdot,+)$ has no negative minimum on $(L,M)$, so to show that $g(\cdot,+)$ is 
non-negative on $[L,M]$ it is sufficient to show that $g(L,-)\geq 0$. This is immediate in case C1, since in this case $g(L,-)=0$. It 
remains to show that 
\begin{equation}\label{max}
g(x,-)\geq 0\text{ on } [\hat m, H].
\end{equation}

We now define $\phi$ by $\phi(x)=rx-\mu_-(x)$, define $\g:[\hat m,H]\rightarrow \R$ by $\g:x\mapsto g(x,-)+\epsilon x^{\frac{3}{2}}$, and note that
$\cL^-\g-r\g=\phi(x) +\epsilon\bigl(\frac{3}{8} \sigma^2_-(x) x^{-\frac{1}{2}}+\half r x^{\frac{3}{2}}-{\frac{3}{2}}\phi(x) x^{\frac{1}{2}}\bigr)\geq (1-{\frac{3}{2}}\epsilon  x^{\frac{1}{2}})\phi(x).$
Thus, taking $0<\epsilon<\frac{2}{3}M^{-\frac{1}{2}}$, $\cL^-\g-r\g\geq 0$ on $(\hat m,H)$, so, by the strong maximum principle, $\g$ has no positive maximum on $(\hat m, H)$.
 Moreover, $\g(\hat m)=\epsilon \hat m^{\frac{3}{2}}>0$ and 
 $\frac{\partial\g}{\partial x}|_{x=\hat m}=\epsilon \hat m^{\frac{1}{2}}>0$ unless $\hat m=0$ in which case $\g(0)\geq 0$, $\g'(0)=0$ and $\frac{\partial^2\g}{\partial x^2}|_{x=\hat m}=\infty$. In either case, $\g$ is initially {\em strictly} increasing and non-negative
so must be monotone increasing on $[\hat m,H]$. We conclude that for each positive $\epsilon$, $\g$ is non-negative and monotone increasing on $[\hat m,H]$ and so, taking the limit as $\epsilon\rightarrow 0$ we conclude that $g(x,-)$ is increasing on $[\hat m,H]$. This establishes (\ref{max})
\end{proof}

\end{document}